\documentclass[sigconf,nonacm]{acmart}

\usepackage{booktabs}
\usepackage{amsmath}
\usepackage{algorithm}
\usepackage{algpseudocode}
\usepackage{tikz}

\newcommand{\sysname}{TorchIsotonic}
\newcommand{\availability}{\footnote{\url{https://github.com/LenaBarretta/torchisotonic}; the
code of the experiments is in its \texttt{paper} folder.}}

\begin{document}

\title{Merge Hulls Without Copying Vertices: GPU Isotonic Regression}

\author{Elena Raikova}
\affiliation{%
  \institution{Independent researcher}
  \city{Buenos Aires}
  \country{Argentina}}
\email{lenabarretta@gmail.com}

\begin{abstract}
Isotonic regression finds the monotone sequence closest to the data in weighted least
squares. It is used to calibrate probabilities and inside differentiable ranking. Its standard
algorithm, pool adjacent violators, is sequential within a sequence, and torchsort, the exact
GPU solver of differentiable sorting libraries, runs it with one thread per sequence, which is
slow on long ones. The fit is the slope sequence of the lower
convex hull of cumulative sums, and hulls of sorted points can be merged in parallel. Merging
partial hulls level by level, however, copies the surviving vertices at every level, which
moves $\Theta(n \log n)$ data when most points survive. We store each chunk's hull once and
represent a merged hull as a list of segments of these chunk hulls, at most one per chunk; a
merge then updates the list. We prove the representation correct, bound the work and depth of
the merge phase, and measure a CUDA implementation on four NVIDIA GPUs against a control that
copies the surviving vertices with the whole GPU at every level. On dense hulls the merge
phase is 17 to 276 times faster with segment lists and the whole operator 3.9 to 19 times; on
noisy data, where few vertices survive, the difference is small. In a differentiable soft rank,
replacing torchsort's isotonic solver makes forward and backward 1.8 to 820 times faster, and
the training step of a compact model with a Spearman loss 1.5 to 40 times.
\end{abstract}



\maketitle

\section{Introduction}
Click models in online advertising are often trained on data in which the non-clicks were
downsampled, and their predicted probabilities are then too high~\cite{he2014}. Before the
predictions are used to price impressions, they are calibrated: a monotone curve, fitted to a
day of held-out traffic, maps each score to the click rate observed at that score. Isotonic
regression fits this curve~\cite{zadrozny2002,niculescu2005}. A day of traffic has millions of
scores, and calibration per campaign or per country needs hundreds of curves. The same
problem, the non-decreasing sequence closest to the data in weighted least squares, appears in
differentiable sorting and ranking~\cite{blondel2020}, sparse top-$k$
operators~\cite{sander2023} and distributional regression~\cite{henzi2022}, usually as a batch
of many rows.

The standard algorithm, pool adjacent violators, scans a sequence once and merges neighbouring
blocks while they violate the order. torchsort~\cite{torchsort}, the exact GPU solver used by
differentiable sorting libraries, runs it with one thread per sequence. This is fast for many
short rows. For one row of $10^7$ elements
such a solver takes $5.5$~s on a T4 and $6.4$~s on an A100, while SciPy takes $0.17$~s on one
CPU core. Iterative methods for accelerators~\cite{sander2023} are parallel but approximate.

The fit is also the slope sequence of the lower convex hull of the cumulative sum
diagram~\cite{barlow1972}. Lower hulls of points sorted by abscissa can be built in parallel:
split the points, build the hulls of the parts, and merge neighbouring hulls along the bridge
between them~\cite{preparata1977,overmars1981}, with logarithmic depth on a
PRAM~\cite{goodrich1987,berkman1996}. On a GPU the direct version merges partial hulls level
by level and copies the surviving vertices at every merge. On noisy data few vertices survive
and the copies are cheap. On data that are nearly monotone almost every point is a vertex, and
the merge tree moves $\Theta(n \log n)$ vertices.

In our representation each chunk of the input writes its hull's vertices once, and a merged
hull is a list of \emph{segments}, ranges of vertices of the chunk hulls with at most one per
chunk. A merge finds the bridge, shortens the two segments it touches, and appends the right
list to the left one (Section~\ref{sec:segment-lists}). Structures that merge hulls without
copying exist for pointer machines and PRAMs (Section~\ref{sec:related}); ours is a flat array
with one descriptor slot per chunk. The contributions are:
\begin{itemize}
\item \textbf{Segment-referenced hull merging} (Sections~\ref{sec:algorithm}
and~\ref{sec:correctness}): the data structure, the merge, and a proof that segment lists
represent every hull of the merge tree (Theorem~\ref{thm:segments}).
\item \textbf{Bounds} on the work and depth of the merge phase (Section~\ref{sec:analysis}),
for the algorithm and for the code as executed.
\item \textbf{A GPU implementation}, \sysname,\availability{} a weighted, differentiable PyTorch operator
for batches of rows of any common length, measured from rows of $16$ elements to a single row
of $10^8$ (Section~\ref{sec:implementation}). On our tests
its float64 fits agree with a float64 reference to $1.3 \cdot 10^{-10}$ and its float32 fits
to $1.2 \cdot 10^{-7}$, in the largest error relative to $1 + |z^\star_i|$
(Section~\ref{sec:numerics}).
\item \textbf{An evaluation} on T4, L4, A100 and H100 GPUs (Section~\ref{sec:evaluation}):
segment lists against two controls that copy the surviving vertices, ablations of the other
design choices, comparisons with a one-thread-per-row solver and with scikit-learn, a differentiable soft
rank and a training step that uses it, and the calibration of a click model on the Criteo
data.
\end{itemize}

\section{Background}
\label{sec:background}

\subsection{Isotonic regression as a lower convex hull}
\label{sec:gcm}

Given $y_1, \dots, y_n \in \mathbb{R}$ and weights $w_1, \dots, w_n > 0$, the (increasing)
isotonic regression problem is
\begin{equation}
  \min_{z_1 \le z_2 \le \dots \le z_n} \sum_{i=1}^{n} w_i (z_i - y_i)^2 .
  \label{eq:iso}
\end{equation}
Its solution $z^*$ is unique and piecewise constant: it consists of \emph{blocks} of
consecutive indices on which $z^*$ equals the weighted mean of $y$ over the block. The
decreasing problem is the increasing one for $-y$.

The classical geometric characterisation (Barlow et al.~\cite{barlow1972}; Robertson et
al.~\cite{robertson1988}, Theorem~1.2.1) is the one this paper computes with. Let the
\emph{cumulative sum diagram} be the points
\begin{equation}
  P_k = (W_k, S_k), \qquad W_k = \sum_{i \le k} w_i, \quad S_k = \sum_{i \le k} w_i y_i,
  \qquad k = 0, \dots, n,
\end{equation}
with $P_0 = (0, 0)$; their abscissae strictly increase since the weights are positive. The
greatest convex minorant of the diagram is its lower convex hull, and
\begin{equation}
  z^*_i = \text{the slope of the hull edge above } [W_{i-1}, W_i] .
  \label{eq:slope}
\end{equation}
Indeed, the slope of the edge from vertex $P_a$ to vertex $P_b$ equals
$(S_b - S_a)/(W_b - W_a)$, the weighted mean of $y$ over $(a, b]$, and the vertices are
exactly the block boundaries. The sequential pool-adjacent-violators algorithm
(PAVA) builds this hull as the monotone chain does: pooling two blocks removes a non-convex
vertex.

Throughout, we take a point lying on the segment between its neighbours \emph{not} to be a
vertex (collinear points are removed). This changes which points are called vertices but
not the slopes in~\eqref{eq:slope}, hence not $z^*$. The lane and chunk hulls and the walks
follow it; the bridge search of the merge levels may keep a collinear point
(Section~\ref{sec:correctness}), which changes no slope either.

\subsection{Parallel hulls of presorted points}
\label{sec:parallel-hulls}

The lower hull of $n$ points sorted by abscissa can be built in parallel with $O(n)$ work and
polylogarithmic depth by merging the hulls of neighbouring ranges of points
(Section~\ref{sec:related}). Two $x$-separated lower hulls merge along their \emph{bridge}, the
common supporting line below both. Preparata and Hong~\cite{preparata1977} find it by walking
from the facing ends, Overmars and van Leeuwen~\cite{overmars1981} by a simultaneous binary
search over both hulls in $O(\log n)$ steps, and G{\ae}de et al.~\cite{gaede2024} reduce that
search to three cases on edges. We apply these merges to the cumulative sum diagram on a GPU.

\section{The Algorithm}
\label{sec:algorithm}
The input is a batch of rows, each an independent problem with its cumulative sum diagram
(Section~\ref{sec:gcm}). The algorithm works at three levels. A \emph{lane} of a warp builds
the hull of 32 consecutive points. The lanes of a \emph{chunk} (one to four warps, $C = 1024$
to $4096$ points) merge their hulls in shared memory. The chunks of a row merge in a tree, in
global memory, as lists of segments. A row that fits in one chunk is finished at the second
level, and rows of up to 512 elements are packed several to a warp.

\subsection{Lane hulls in a bit mask}
Each lane takes 32 consecutive elements, forms their cumulative sums in shared memory, and runs
the monotone chain over its points. The stack of the chain is a 32-bit mask in a register, with
bit $k$ set while point $k$ is on the stack. The top and the entry under it are kept in
registers, and the entry under those is found with a count-leading-zeros instruction on the
mask. A pop clears a bit and a push sets one. Every iteration either pops or pushes and issues
both loads it might need, so the loop has no data-dependent branch and the 32 lanes of a warp
stay in step. When the chain ends, the set bits are the lane's hull.

\subsection{Merging lane hulls by walking}
The lane hulls of a chunk merge pairwise in rounds, five within a warp and then across warps.
Neighbouring hulls share their boundary point. The first lane of a merge tests whether the
chain is convex at the shared point. If it is not, the lane walks from the point's two
neighbours, moving the left end left and the right end right until the line between them
supports both hulls; this line is the bridge. Every lane of the merge then clears the mask bits
strictly between the bridge's ends. Each step of a walk except the last of each pass removes a
vertex, and a vertex is removed only once, so all the walks of a chunk take $O(C)$ steps
(Lemma~\ref{lem:walk}). The chunk's hull is then the set bits of its lanes' masks.

\subsection{Chunk hulls as segment lists}
\label{sec:segment-lists}
A chunk that is not a whole row writes its hull to a slot of $C + 1$ entries of a global
vertex array. The vertices are ranked by a prefix sum of the lanes' mask counts and stored as
coordinates relative to the chunk's first point, with the point's index. These entries are
written once; the merges above change only descriptors. The chunk's \emph{segment list} is one
descriptor $(j, \ell, h, v)$ of 24 bytes: the chunk $j$, its first and last surviving vertex
$\ell$ and $h$, and the slot $v$ that holds its vertices. A node of the merge tree spanning
chunks $j_0, \dots, j_1$ is the list of the segments of its hull, at most one per chunk
(Theorem~\ref{thm:segments}). The list is stored in the descriptor slots of the node's chunks,
starting at $j_0$'s (Figure~\ref{fig:segments}).

\begin{algorithm}[t]
\caption{Merging the segment lists $L$ and $R$ of two neighbouring nodes.}
\label{alg:merge}
\begin{algorithmic}[1]
\State $s \gets$ the last vertex of $L$ (the first of $R$)
\If{the chain is convex at $s$}
  \State append $R$, without $s$, to $L$ \Comment{no bridge}
\Else
  \State $(a, b) \gets$ walk from $s$'s neighbours, $\le T$ moves \Comment{thread levels only}
  \If{there was no walk, or it gave up} \Comment{bisection}
    \State $a \gets$ first vertex of $L$ above its tangent to $R$
    \State $b \gets$ where that tangent touches $R$
  \EndIf
  \State end $L$ at $a$: cut its segment, drop the later ones
  \State start $R$ at $b$: cut its segment, drop the earlier ones
  \State move $R$'s descriptors to follow $L$'s
\EndIf
\end{algorithmic}
\end{algorithm}

Level $r$ of the tree merges pairs of nodes of $2^r$ chunks (Algorithm~\ref{alg:merge}). A
vertex is addressed by its segment and its position in the chunk's slot. Its coordinates in the
row are its chunk-relative ones plus the chunk's start, which comes from a prefix sum over the
chunks' totals. The bridge is found by walking, as within the chunks, now across segment
boundaries. A walk that would take more than $T = 64$ steps is replaced by a search: a
bisection over the candidate vertices of $L$, each tested against its tangent to $R$, both
bisecting first over segments and then within one. The merge then cuts two segments and moves
the descriptors of $R$'s tail down by at most $2^r$ places. A level with many merges runs one
thread per merge, since its lists are short and a walk is serial. A level with few merges runs
one warp per merge, which searches at once with 32 candidates per step and copies descriptors
in coalesced rounds.

\subsection{Reconstruction}
At the root, the list holds the row's hull. One kernel finds, for every chunk, the range of
its surviving vertices and the slopes of the edges into and out of it, by a bisection over the
root list. A chunk without a segment lies under a single edge. A second kernel assigns each of
the chunk's elements the slope of the edge above it, by a binary search in the span's
vertices, and writes the fit.

\subsection{Short rows}
Rows of up to 512 elements are packed several to a warp. A row takes the smallest power of two
$L$ of lanes with $32L \ge n$, and these lanes run the lane and chunk levels on their own, with
the prefix sums and merges restricted to the row. The whole row stays in shared memory and
is finished in one kernel.

\begin{figure}[t]
\centering
\begin{tikzpicture}[x=0.33cm, y=0.33cm, font=\scriptsize,
  v/.style={draw, minimum size=0.3cm, inner sep=0pt},
  live/.style={v, fill=black!55},
  dead/.style={v, fill=black!8},
  copied/.style={v, fill=black!30},
  seg/.style={draw, rounded corners=1pt, inner sep=1.2pt, fill=white}]
  \node[anchor=west] at (-0.5, 9.9) {\textbf{(a)} merging by copying};
  \foreach \j/\n in {0/6, 1/5, 2/6, 3/4} {
    \foreach \k in {1,...,\n} \node[live] at (\j*6.2 + \k*0.85 - 0.85, 8) {};
    \node at (\j*6.2 + 1.6, 8.9) {chunk \j};
  }
  \foreach \k in {1,...,7} \node[copied] at (\k*0.85 + 1.2, 5.6) {};
  \foreach \k in {1,...,8} \node[copied] at (\k*0.85 + 13.2, 5.6) {};
  \foreach \k in {1,...,8} \node[copied] at (\k*0.85 + 6.9, 4.0) {};
  \draw[->, thin] (2.2, 7.5) -- (3.4, 6.1);  \draw[->, thin] (7.8, 7.5) -- (6.0, 6.1);
  \draw[->, thin] (14.6, 7.5) -- (15.4, 6.1); \draw[->, thin] (20.0, 7.5) -- (18.8, 6.1);
  \draw[->, thin] (4.6, 5.1) -- (8.4, 4.5);   \draw[->, thin] (16.4, 5.1) -- (12.6, 4.5);
  \node[anchor=west] at (22.4, 5.6) {level 1};
  \node[anchor=west] at (22.4, 4.0) {level 2};
  \node[anchor=west] at (-0.5, 2.4) {\textbf{(b)} merging segment lists};
  \foreach \j/\n/\lo/\hi in {0/6/0/3, 1/5/2/3, 2/6/9/9, 3/4/2/3} {
    \foreach \k in {0,...,5} {
      \ifnum\k<\n
        \pgfmathtruncatemacro{\on}{(\k>=\lo) && (\k<=\hi) ? 1 : 0}
        \ifnum\on=1 \node[live] at (\j*6.2 + \k*0.85, 1.2) {};
        \else \node[dead] at (\j*6.2 + \k*0.85, 1.2) {};
        \fi
      \fi
    }
    \node at (\j*6.2 + 1.6, 0.2) {chunk \j{} (fixed)};
  }
  \node[seg, anchor=west] at (-0.5, -1.0) {$(0, 0, 5)$};
  \node[seg, anchor=west] at (5.7, -1.0) {$(1, 0, 4)$};
  \node[seg, anchor=west] at (11.9, -1.0) {$(2, 0, 5)$};
  \node[seg, anchor=west] at (18.1, -1.0) {$(3, 0, 3)$};
  \node[seg, anchor=west] at (-0.5, -2.4) {$(0, 0, 3)\ (1, 2, 4)$};
  \node[seg, anchor=west] at (11.9, -2.4) {$(2, 0, 4)\ (3, 1, 3)$};
  \node[seg, anchor=west] at (-0.5, -3.8) {$(0, 0, 3)\ (1, 2, 3)\ (3, 2, 3)$};
  \node[anchor=west] at (22.4, -1.0) {level 0};
  \node[anchor=west] at (22.4, -2.4) {level 1};
  \node[anchor=west] at (22.4, -3.8) {level 2};
\end{tikzpicture}
\caption{Merging four chunk hulls. (a) Every merge copies the surviving vertices into a new
array, so a vertex that survives $k$ levels is written $k + 1$ times. (b) Each chunk writes
its vertices once; a node is a list of segments (chunk, first vertex, last vertex), and a
merge cuts the two segments its bridge lands in and moves the right list's descriptors. At
level 2 the bridge passes over chunk 2, which keeps no segment. Dark: vertices of the final
hull.}
\label{fig:segments}
\end{figure}
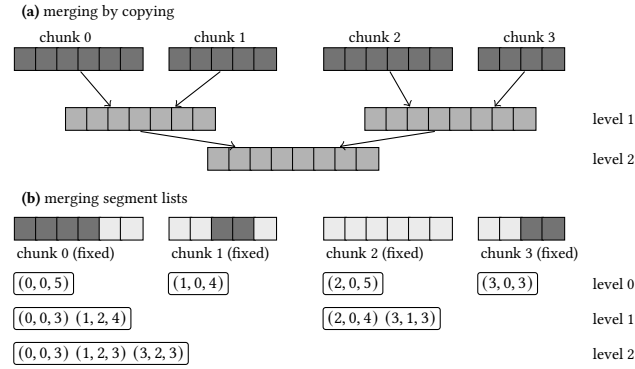

\section{Correctness}
\label{sec:correctness}

All statements are in exact arithmetic; rounding is the subject of
Section~\ref{sec:numerics}. Let $P_0, \dots, P_n$ be the cumulative sum diagram
(Section~\ref{sec:gcm}), split into \emph{chunks}: chunk $j$ covers the points
$P_{c_j}, \dots, P_{c_{j+1}}$, so neighbouring chunks share a point. Within a chunk the same
holds for its lanes. We write $H(Q)$ for the lower hull of a point set $Q$, as its vertex
sequence by increasing abscissa.

\begin{lemma}[filtering]
\label{lem:filter}
Let $Q \subseteq \{P_0, \dots, P_n\}$. If $P_k \in Q$ is a vertex of $H(\{P_0, \dots, P_n\})$
then $P_k$ is a vertex of $H(Q)$.
\end{lemma}
\begin{proof}
$P_k$ is a vertex of the lower hull of a set exactly when some non-vertical line through
$P_k$ has every point of the set on or above it, with no two points of the set on the line on
both sides of $P_k$. The line that witnesses this for the whole diagram witnesses it for the
subset $Q$.
\end{proof}

Hence a point that is not a vertex of its lane's hull, or of its chunk's, is not a vertex of
the diagram's hull, and every stage may discard the non-vertices of its own level. The end
points of a lane, chunk or node are always vertices of its hull, since they are extreme in
$x$.

\begin{lemma}[merging, \cite{preparata1977,overmars1981}]
\label{lem:merge}
Let $A$ and $B$ be point sets with $x(a) \le x_s \le x(b)$ for all $a \in A$, $b \in B$,
sharing the point $s$ with $x(s) = x_s$, which is the last vertex of $H(A)$ and the first of
$H(B)$. Then $H(A \cup B) = H(A)[0..\alpha] \,\|\, H(B)[\beta..]$ for a unique bridge
$(H(A)[\alpha], H(B)[\beta])$: if the chain is convex at $s$ (the slope into $s$ along $H(A)$
is below the slope out of $s$ along $H(B)$) then $\alpha$ and $\beta$ both denote $s$, which
is kept once; otherwise $s$ is not a vertex of $H(A \cup B)$, $\alpha$ precedes $s$ and
$\beta$ follows it, and the vertices of $H(A)$ after $\alpha$ and of $H(B)$ before $\beta$
lie on or above the bridge.
\end{lemma}

The implementation finds the bridge by walking from $s$: $\alpha$ moves left while the turn
$(H(A)[\alpha-1], H(A)[\alpha], H(B)[\beta])$ is not strictly left, $\beta$ moves right
symmetrically, alternating until neither moves. The fixed point supports both chains, so it
is the bridge. Every step except the last of each pass removes a vertex, and a vertex is
removed once over all merges of a node, which bounds the walks of a node by $O$(its
vertices) in total (Section~\ref{sec:analysis}).

\begin{definition}[segment lists]
A \emph{segment} $(j, \ell, h)$, $\ell \le h$, denotes the vertices $\ell..h$ of chunk $j$'s
hull $H_j$. A \emph{segment list} is a sequence of segments with strictly increasing chunk
indices; it denotes the concatenation of its segments' vertices.
\end{definition}

\begin{theorem}[segment representation]
\label{thm:segments}
For every node of the merge tree, the vertex sequence of the hull of the node's points, without
collinear middle points (Section~\ref{sec:gcm}), equals the concatenation of its segment list,
and each segment is a contiguous range of one chunk hull's vertices. The list of a node of $c$
chunks has at most $c$ segments.
\end{theorem}
\begin{proof}
By induction up the tree. A leaf, chunk $j$, has the list $[(j, 0, |H_j| - 1)]$. Let nodes
$A$ and $B$ be neighbours with lists $S_A$, $S_B$ satisfying the claim; their hulls share the
point $s$ at the end of $S_A$ and the start of $S_B$, and $A$ lies to the left of $s$, $B$ to
the right. By Lemma~\ref{lem:merge}, $H(A \cup B) = H(A)[0..\alpha] \,\|\, H(B)[\beta..]$.
The prefix $H(A)[0..\alpha]$ of the concatenation of $S_A$ is the sequence of segments of
$S_A$ up to the one containing $\alpha$, that one cut at $\alpha$'s position, $(j, \ell,
\mathrm{pos}_\alpha)$ with $\ell \le \mathrm{pos}_\alpha$; the suffix $H(B)[\beta..]$ is the
segment of $S_B$ containing $\beta$, cut to $(j', \mathrm{pos}_\beta, h)$ with
$\mathrm{pos}_\beta \le h$, followed by the rest of $S_B$. Cut segments remain contiguous
ranges of one chunk hull, and the chunk indices of $S_A$ are below those of $S_B$, so the
concatenation is a segment list of $A \cup B$. Since the kept part of $S_A$ and of $S_B$ use
disjoint chunks, the list has at most as many segments as the node has chunks.
\end{proof}

By the last clause, a node's list fits in one slot per chunk of the node. A merge can
therefore write its result where the left child's list was, moving the right child's tail down
by at most the number of its segments.

\begin{corollary}[reconstruction]
\label{cor:fit}
At the root, the segment list denotes $H(\{P_0, \dots, P_n\})$. For chunk $j$ with a
segment $(j, \ell, h)$ in it, the edges above the chunk's elements are the edges between
consecutive vertices $\ell..h$, the edge into vertex $\ell$ from the previous segment's last
vertex, and the edge out of $h$ into the next segment's first vertex; a chunk with no segment
lies under a single edge, from the last vertex of the last segment of an earlier chunk to the
first vertex of the next segment. Assigning each element the slope of the edge over it gives
$z^*$ by~\eqref{eq:slope}.
\end{corollary}

Degenerate cases follow the same convention. A node always has at least two vertices, its
ends, so the predecessor of $s$ in $H(A)$ and its successor in $H(B)$ exist. The shared point
appears in both lists; it is kept once in the convex case and dropped from both otherwise. A
cut cannot produce an empty segment. A turn of zero is treated as non-convex in the lane and
chunk hulls and in the walks, so collinear middle points are dropped there. Theorem~\ref{thm:segments} is
stated for this canonical hull. The implementation's bridge search may stop at the first of
several collinear points of contact and so keep a collinear middle point in a merged list. Such
a list then also contains points that lie on edges of the hull. A point on an edge changes no
slope, so the fit is the same.

\section{Analysis}
\label{sec:analysis}
We analyse one row of $n$ elements; rows of a batch are independent. Let $C$ be the chunk
size ($C = 1024$ for long rows), $m = \lceil n/C \rceil$ the number of chunks, $h_j$ the
number of vertices of chunk $j$'s hull and $H = \sum_j h_j \le n + m$, and $P = 32$ the points
per lane. The merge tree has $R = \lceil \log_2 m \rceil$ levels; a node of level $r$ spans
$2^r$ chunks and, by Theorem~\ref{thm:segments}, its list has at most $2^r$ segments. Work
counts operations over all threads; depth counts the longest chain of dependent operations.

\subsection{Chunk hulls}
A lane processes its $P$ points with one push and at most one pop each, in a loop whose only
data-dependent branch is the rare fallback of the turn predicate, so the lane hulls of a chunk
cost $\Theta(C)$ operations at depth $\Theta(P)$. The lane and warp merges inside a chunk walk
to their bridges:

\begin{lemma}[walks]
\label{lem:walk}
In a merge by walking (Section~\ref{sec:correctness}), every step of a pass except its last
removes a vertex of one of the two children's hulls, one that is not a vertex of the merged
hull, and the outer loop ends after the first pass without a step. Hence a merge that removes
$d$ vertices takes at most $3d + 2$ steps, and over all merges of a node, whose vertices are
each removed at most once, the walks take $O$(vertices of the node's children) steps.
\end{lemma}

Thus a chunk's merges take $O(C)$ steps in total. A step may also skip lanes whose masks are
empty to reach the next vertex, so it costs up to $O(C/P)$ operations, and the chunk's merges
at most $O(C^2/P)$. Their depth is that of the longest single walk, $\Theta(C)$ in the worst
case, when one point at the end of the chunk pulls the bridge across all of it. $C$ is a
constant of the implementation. For fixed $C$ the chunk phase therefore costs
$\Theta(n)$ work at depth $O(C)$, and the exclusive scan of the chunk totals $\Theta(m)$ work
at depth $O(\log m)$.

\subsection{Merging the chunk hulls}
\label{sec:analysis-merge}
A merge at level $r$ tests the shared point, finds the bridge, and moves the right child's
tail of at most $2^r$ descriptors down. Moving costs $O(m)$ descriptors per level, hence
$O(m \log m)$ in total. This is the only data movement of the phase that grows with $n$, and it
depends on the number of chunks only. For a fixed $C$, $m \log m = (n/C) \log (n/C)$ is not
$O(n)$. It is below $n$ when $C \ge \log_2 (n/C)$, which holds for $C = 1024$ on every input
that fits in memory; the statements below keep the term.

By Lemma~\ref{lem:walk}, finding the bridges by walking alone is $O(H)$ work over the whole
tree. Its depth, however, is not bounded by any polylogarithm. In a row whose last point lies
far below the rest, the merges along the right spine remove all other vertices, and the walk
at the top level takes $\Theta(n)$ steps. The implementation therefore caps each walk and then
searches: a candidate vertex of the left list is tested against its tangent to the right list,
and both are found by bisection, first over segments and then within one.

\begin{theorem}[merge phase, nested search]
\label{thm:merge-impl}
In the work-depth model, with the descriptor moves of a merge done as one parallel copy:
with every walk capped at a constant $T$ steps ($T = 64$ in the implementation) and the
nested search beyond it, the merge phase over $m = \lceil n/C \rceil$ chunk hulls does
$O(n + m \log m)$ work at depth $O(T \log n + \log^3 n)$.
\end{theorem}
\begin{proof}
Capped walks take at most $3H + 2m = O(n)$ steps in total (Lemma~\ref{lem:walk}). A merge
searches at most once. At level $r$ there are at most $\lceil m / 2^{r+1} \rceil$ merges, and
the outer bisection of a search takes $O(r + \log C)$ candidates, each tested with a tangent
of $O(r + \log C)$ steps; the searches therefore cost
$\sum_r (m / 2^{r+1} + 1)\, O((r + \log C)^2) = O(m \log^2 C + \log^3 m)$, which is $O(n)$, as
$m \log^2 C = n \log^2 C / C$. The descriptor moves cost $O(m \log m)$. A level's depth is at
most $T + O((r + \log C)^2)$, and the sum over the $\lceil \log_2 m \rceil$ levels is
$O(T \log m + \log^3 n)$.
\end{proof}

Theorem~\ref{thm:merge-impl} assumes that a merge moves its descriptors in one parallel step.
In the code, a merge run by one thread moves them one by one, and a merge run by a warp moves
them in rounds of 32. This changes the depth.

\begin{proposition}[merge phase, as executed]
\label{prop:merge-code}
The CUDA code's merge phase does $O(n + m \log m)$ work at depth
$O(T \log n + \log^3 n + m)$: the descriptor moves add at most $m/16 + \log_2 m$ sequential
rounds to the bound of Theorem~\ref{thm:merge-impl}.
\end{proposition}
\begin{proof}
A level runs one thread per merge only if it has at least 256 merges, so its nodes span
$2^r \le m/512$ chunks, and a thread moves at most $2^r$ descriptors; over these levels that is
at most $m/256$ steps. A level with fewer merges runs a warp per merge, which moves at most
$2^r$ descriptors in $\lceil 2^r/32 \rceil$ rounds; over all levels at most
$\sum_{r < \lceil \log_2 m \rceil} (2^r/32 + 1) < m/16 + \log_2 m$ rounds. The warp levels do not
walk, and their search is the nested one with 32 candidates per step, so their depth is at
most that of Theorem~\ref{thm:merge-impl}.
\end{proof}

The added term is linear in $n$ with a small constant, $m/16 = n/16\,384$. For $n = 10^8$,
$m = 97\,657$, and the moves take at most about $6\,100$ rounds. The term is largest on dense hulls, where every chunk
keeps a segment. The measurements of Section~\ref{sec:evaluation} include it.

The cube in the depth comes from the nested search, a bisection inside a bisection. The
simultaneous bisection of both chains of Overmars and van Leeuwen~\cite{overmars1981}, in the
three-case form of G{\ae}de et al.~\cite{gaede2024}, finds a bridge in a logarithmic number of
steps and lowers the depth of the model to $O(T \log n + \log^2 n)$ (Appendix~\ref{app:three},
Theorem~\ref{thm:merge}). We implemented it as a variant. It is slower in practice
(Section~\ref{sec:ablations}), so the default keeps the nested search. Since the levels run one
after another, neither search reaches the $O(\log n)$ depth of Goodrich's CREW
hull~\cite{goodrich1987}.

\subsection{Reconstruction}
The ends kernel locates each chunk's segments by bisection in the root list, which is
$O(m \log m)$ work. The fit kernel loads at most $h_j + 2$ slopes per 1024 elements and finds
each element's edge by a binary search of $O(\log h_j)$ steps, $h_j \le C + 1$. The phase
costs $O(n \log C)$ work in the worst case and $\Theta(n)$ for small hulls, at depth
$O(\log n)$.

\subsection{Totals, data movement and memory}
Over all phases the work is $O(n \log C + m \log m)$, the searches included. The fit
contributes a factor $\log C$ per element and the descriptor moves $\log m / C$. The depth is
$O(C + \log^3 n)$ in the model and $O(C + \log^3 n + m)$ as executed
(Proposition~\ref{prop:merge-code}). The walks and searches depend on the hull's geometry in
either representation. The difference lies in the data moved. With materialized hulls a merge
at level $r$ copies the right child's surviving vertices, up to all of them, so the phase moves
$\Theta(H \log m)$ vertices of 12 to 20 bytes, which is $\Theta(n \log n)$ when all points are
vertices. With segment lists it moves $O(m \log m)$ descriptors of 24 bytes, 40~MB for
$n = 10^8$. Peak memory is dominated by the vertex slots, $C + 1$ per chunk of 12 bytes (20 with
weights). This is three times the size of a float32 input, or 2.5 times that of the values and
weights together. The lists, counts and chunk starts take $O(m)$.

\section{Implementation}
\label{sec:implementation}
The operator has one CUDA code path for float32 and float64 and for weighted and unweighted
problems. A call launches at most five kernels plus one per merge level. It also fills the
rows' NaN flags and, for weighted calls, checks that the weights are positive, which waits for
the GPU. The constants below were tuned on the T4 (Section~\ref{sec:evaluation}).

\paragraph{Blocks and memory.}
Rows of up to 512 elements are packed several per warp and solved in registers and shared
memory. Rows of up to 4\,096 elements are one chunk, handled by a block of one, two or four
warps; for weighted float64 the limit is 2\,048, as a block of four warps would not fit in
shared memory. Longer rows are split into chunks of 1\,024 elements, one warp per chunk. The
block size is the largest whose shared memory (the chunk's cumulative sums, lane starts and
vertex stack) fits in the 48~KB of a launch. The lane hulls live in registers and use no local
memory. A chunk's vertex slot takes 12 bytes per vertex, 20 with weights.

\paragraph{Merge levels.}
The $\lceil \log_2 m \rceil$ levels are separate launches. A level runs one thread per merge if
it has at least 256 merges, and one warp per merge otherwise. A thread walks for at most 64
steps before it switches to the nested search; a warp searches directly, 32 candidates per
step, and moves the tail of the right list with coalesced copies. The three-case search of
Theorem~\ref{thm:merge} is implemented as a variant for the thread-per-merge levels.

\paragraph{Coordinates.}
The coordinates form a hierarchy, so that differences of nearby points keep their precision.
Within a lane, cumulative sums are taken in the working type from the lane's first point. In
float32 they are pairs (hi, lo) from a renormalized two-sum. The lo part is stored in shared
memory as a 16-bit code in units of $2^{-38}$ of the power of two of hi, which suffices since
a renormalized pair bounds it by $2^{14}$ such units; a pair then takes 6 bytes instead of 8 and
resolves $2^{-38}$ of a sum's magnitude. Lane starts within a chunk are pairs from a two-sum
scan across the warp. A chunk's vertices are stored as doubles relative to the chunk, and the
chunk starts along a row are double-double pairs from a scan over the chunk totals. A slope
between vertices of different chunks is formed from the differences of the hi parts (exact for
nearby chunks, by Sterbenz's lemma), of the lo parts and of the chunk-relative coordinates.
Absolute coordinates reach $10^{10}$ at $10^8$ elements, and a difference of two of them would
lose six digits.

\paragraph{Predicates.}
The turns of the lane and chunk levels are first evaluated on the hi parts. If the result
exceeds an a-priori error bound it is used; otherwise the turn is recomputed from the pairs
with two-sum and two-product. The bound is in terms of the magnitudes of the summed terms, the
lane's start and the sum within the lane, rather than of the coordinate alone, since the two
may cancel. The merge levels evaluate the cross product from absolute coordinates, one
addition per point, and recompute it from the differences above only if its magnitude is
within an error bound. That recomputation is double arithmetic at the scale of a chunk's sums.
On the T4 the merge levels are limited by double throughput, a thirty-second of the float
rate. The cross product takes seven double operations, and evaluating its bound in double, or
in float with conversions that run at the double rate, doubles the time of the merge levels.
The code therefore replaces each term of the bound by a power of two read from the exponent
field of a double's high word, including the chunk starts, and compares exponents in integer
arithmetic. The test is coarser than the bound by a constant factor, and zero or subnormal
operands always take the second evaluation. Filters of this kind are
classical~\cite{shewchuk1997,bartels2023}; this one uses no floating-point operation.
Appendix~\ref{app:filter} gives the bounds and their derivation, and
Section~\ref{sec:numerics} the accuracy they lead to.

\paragraph{Output.}
Every slope of the output is rounded to the working type once: in float32 from a quotient of
the pairs corrected by its remainder, in float64 from double coordinates. The output is not
guaranteed to be monotone (Section~\ref{sec:numerics}). The option \texttt{exact=False} drops
the pairs and the filter of the lane and chunk levels for float32.

\paragraph{Gradient.}
The fit is constant on blocks, each the weighted mean of $y$ over the block, so its Jacobian
is block averaging. With $g$ the gradient with respect to the fit and $G$, $W$ the sums of $g$
and $w$ over an element's block, $\partial/\partial y_i = w_i G / W$ and $\partial/\partial w_i =
(y_i - z^*_i)\, G / W$, wherever the blocks do not change. The blocks are those of the hull.
When a gradient is needed, the forward pass also writes one byte per element, set where the
element's point is a vertex of the row's hull, that is where a block starts. The fit kernel
knows this for every element it writes. An element at the start of a chunk is a vertex when it
closes the previous chunk's segment, which the ends kernel records. The backward pass sums over
these blocks without atomics. A warp takes a tile of the row, sized to the row and at most
1\,024 elements. A row that is one tile is finished in one kernel; a longer row takes three
steps, and the sums of the blocks that cross tiles come from prefix sums over the tiles'
boundary sums. The sums are pairs of the working type. Neighbouring blocks whose exact means
differ by less than the solver resolves may be pooled, so near such ties the gradient follows
the solver's partition; Appendix~\ref{app:grad} describes how the gradients are validated.

\section{Evaluation}
\label{sec:evaluation}
\subsection{Setup}
The measurements are on four NVIDIA GPUs of four generations: a Tesla T4 (Turing) hosted by
Kaggle, on a host with an Intel Xeon at 2.00~GHz (2 cores, 4 threads) and 31~GB of memory,
driver 580.178 and Python 3.13; and an L4 (Ada), an A100-SXM4-40GB (Ampere) and an H100 80GB
HBM3 (Hopper) hosted by Modal, driver 610.57 and Python 3.12. All run CUDA 12.8
and PyTorch 2.11. The baselines are torchsort 0.1.10, built from PyPI against the installed
PyTorch, and, on the T4's host, SciPy 1.16.3 and scikit-learn 1.6.1 with NumPy 2.1.3. We
measure \sysname{} 0.2.1. The constants
of the implementation (chunk sizes, the walk cap, the scheduling threshold) were tuned on the T4 alone; the other three GPUs show whether
the design depends on it.

Inputs are resident on the GPU. A fit is timed with the host's clock around the call and a
\texttt{torch.cuda.synchronize} after it, over ten calls after a warm-up and a busy wait that
raises the GPU's clock, stopping early once the calls have taken ten seconds, so that the
slowest torchsort cells are of two calls. We report the median and, where it matters, the
interquartile range. The gradient benchmark times each call with CUDA events, and the
per-stage times are GPU kernel times from the PyTorch profiler, means over five calls. On
noisy data the per-stage times repeat to within a few percent across sessions, while on dense
data the chunk stage of $10^8$ increasing elements varies by up to 25\% between calls of one
session on the T4. Comparisons between variants are therefore taken within one session, with
the default profiled beside each variant; the differences between representations in
Section~\ref{sec:vs-materialized} are far larger than this variation, while some of the
ablations are not.

The data families determine how many points survive as hull vertices. In \emph{noise}, a
ramp from 0 to 3 plus standard normal noise, few do, and in strictly increasing data all
points are vertices. In a sawtooth the bridges reach across many chunks, and in a
\emph{cliff}, increasing values followed by one far below, they reach across the whole row.
Each number below states its family. The shapes run from $10^6$ rows of $16$ elements through $10^4 \times 10^3$, $10^3 \times 10^4$
and $10^2 \times 10^5$ to one row of $10^7$ and one of $10^8$.

\subsection{Comparison with a one-thread-per-row solver}
\label{sec:vs-torchsort}
\begin{table*}[t]
\caption{Decreasing unweighted float32 fits, data on the GPU, ms, median of ten calls (two for
the slowest torchsort cells). The interquartile range of ours is within 6\% of the median
except at $100 \times 10^3$, where a call takes $0.1$~ms (up to 38\%, on the T4). torchsort runs pool adjacent violators with one thread per row; it was
not run on $10^8$ elements.}
\label{tab:torchsort}
\small
\begin{tabular}{@{}rrrrrrrrr@{}}
\toprule
& \multicolumn{2}{c}{T4 (Turing)} & \multicolumn{2}{c}{L4 (Ada)} & \multicolumn{2}{c}{A100 (Ampere)} & \multicolumn{2}{c}{H100 (Hopper)} \\
batch $\times$ $n$ & ours & torchsort & ours & torchsort & ours & torchsort & ours & torchsort \\
\midrule
$100 \times 10^3$ & 0.12 & 3.04 & 0.11 & 2.42 & 0.10 & 3.09 & 0.07 & 2.42 \\
$10^6 \times 16$ & 2.25 & 72.2 & 0.73 & 6.95 & 0.49 & 3.40 & 0.28 & 1.91 \\
$10^4 \times 10^3$ & 2.30 & 99.8 & 0.65 & 14.0 & 0.53 & 18.7 & 0.35 & 12.5 \\
$10^3 \times 10^4$ & 1.77 & 197 & 0.88 & 139 & 0.58 & 190 & 0.37 & 123 \\
$1 \times 10^6$ & 0.54 & 550 & 0.40 & 415 & 0.29 & 588 & 0.21 & 432 \\
$1 \times 10^7$ & 2.13 & 5\,493 & 1.23 & 4\,947 & 0.83 & 6\,377 & 0.56 & 4\,683 \\
$1 \times 10^8$ & 20.0 & - & 8.38 & - & 4.33 & - & 2.74 & - \\
\bottomrule
\end{tabular}
\end{table*}
\begin{figure}[t]
\centering
\includegraphics[width=\columnwidth]{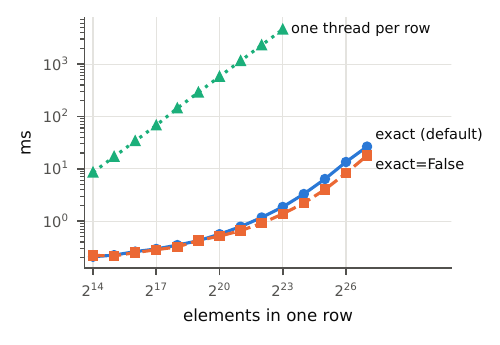}
\caption{One row of $n$ noisy float32 elements on a T4, median of ten calls: the default, the
option without the float32 pairs, and torchsort (to $2^{23}$).}
\label{fig:latency}
\end{figure}
Table~\ref{tab:torchsort} and Figure~\ref{fig:latency} compare the operator with
torchsort~\cite{torchsort}, the exact GPU solver used by differentiable sorting libraries, on
unweighted float32 rows. On many short rows, where a thread per row keeps every thread busy,
torchsort is 7 to 43 times slower. The ratio grows with the row's length and reaches three to
four orders of magnitude for a row of $10^7$ elements. torchsort's time for such a row is about
the same on the H100 as on the T4 and longer on the A100, as a single thread does not benefit
from a larger GPU. With the gradient included the ratio is smaller: $4.2$ against $6\,898$~ms
for a row of $10^7$ elements on the T4 (Section~\ref{sec:gradients}). On the T4's host CPU
(medians), SciPy's pool adjacent violators on one core takes $171$~ms for
a row of $10^7$ elements, and the library's CPU path, parallel over rows on the host's four
threads, $380$~ms for $10^6$ rows of 16; the GPU operator is $80$ and $169$ times faster.

\subsection{Segment-referenced versus materialized merging}
\label{sec:vs-materialized}
\begin{table*}[t]
\caption{The merge stage alone, ms, one row of $n$ float32 elements: segment lists (seg.)
against two controls that copy the surviving vertices of every merge into a flat array, with
the whole GPU at every level (grid) or by the thread or warp of each merge (merge). Same chunk
kernel, bridge search and schedule.}
\label{tab:materialized}
\small
\setlength{\tabcolsep}{4pt}
\begin{tabular}{@{}lrrrrrrrrrrrr@{}}
\toprule
& \multicolumn{3}{c}{T4} & \multicolumn{3}{c}{L4} & \multicolumn{3}{c}{A100} & \multicolumn{3}{c}{H100} \\
data, $n$ & seg. & grid & merge & seg. & grid & merge & seg. & grid & merge & seg. & grid & merge \\
\midrule
increasing, $10^7$ & 0.19 & 25.1 & 171 & 0.15 & 21.6 & 162 & 0.19 & 4.34 & 223 & 0.14 & 2.40 & 173 \\
increasing, $10^8$ & 1.85 & 351 & 2\,350 & 1.17 & 323 & 1\,644 & 1.25 & 50.7 & 2\,175 & 0.95 & 27.5 & 1\,717 \\
cliff, $10^8$ & 2.21 & 343 & 2\,312 & 1.32 & 320 & 1\,611 & 1.32 & 50.2 & 2\,128 & 1.01 & 27.3 & 1\,676 \\
noise, $10^7$ & 0.50 & 0.49 & 0.39 & 0.40 & 0.38 & 0.32 & 0.33 & 0.32 & 0.27 & 0.23 & 0.21 & 0.17 \\
noise, $10^8$ & 2.21 & 3.76 & 3.92 & 0.80 & 1.55 & 0.99 & 0.60 & 0.96 & 0.56 & 0.40 & 0.60 & 0.32 \\
\bottomrule
\end{tabular}
\end{table*}
\begin{figure}[t]
\centering
\includegraphics[width=\columnwidth]{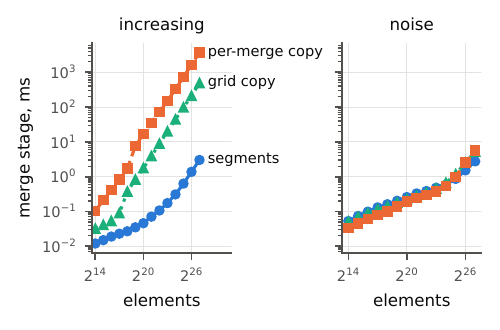}
\caption{The merge stage of one row of $n$ float32 elements on a T4: segment lists against
copying the surviving vertices with the whole GPU (grid copy) and merge by merge, on increasing
data (every point a vertex) and on noise.}
\label{fig:merge}
\end{figure}
Against torchsort the operator is faster on every shape (Section~\ref{sec:vs-torchsort}). This
section isolates the representation, whose gain depends on how many vertices survive the
merges. Table~\ref{tab:materialized} and Figure~\ref{fig:merge} compare it with two
controls. Both keep the pipeline and, after every merge, write the merged hull's vertices out
as one contiguous array. In \emph{per-merge copying}, the thread or warp that found a bridge
also copies the right child's surviving vertices down, in place, as the textbook merge does.
At the top levels of a long row a single warp then copies tens of millions of vertices, so this
control measures copying at the speed of one warp. \emph{Grid copying} finds every merge's
bridge in one kernel and writes every merged hull into another buffer in a second one, with all
of the GPU's blocks and coalesced accesses, as stream compaction would~\cite{billeter2009}.
Each surviving vertex then moves once per level at the memory's bandwidth. Both controls
evaluate their turns on absolute coordinates without the filter of
Section~\ref{sec:implementation}; with plain predicates the default's merge stage is up to
24\% faster (Section~\ref{sec:ablations}).

Where most points survive (increasing data, or a cliff whose bridges reach across the row),
both controls move $\Theta(n \log m)$ vertices. Grid copying moves them at half to two thirds
of each GPU's peak bandwidth. The merge stage with segment lists is $133$ to $190$ times faster
on the T4, $141$ to $276$ times on the L4, $23$ to $41$ times on the A100 and $17$ to $29$ times
on the H100, at $10^7$ and $10^8$ elements. Per-merge copying is another 5 to 72 times slower.
Figure~\ref{fig:merge} shows the growth on the T4: the copies grow as $n \log n$, the segment
lists with the number of chunks. On noisy data only a few vertices of each chunk's hull
survive its first merges. There the three are close at $10^7$ elements, with per-merge
copying up to 27\% faster, as a flat array saves one indirection per vertex. At $10^8$ the
segment lists are $1.5$ to $2.0$ times faster than grid copying, and between $1.8$ times faster
(T4) and $23\%$ slower (H100) than per-merge copying.

\begin{table*}[t]
\caption{The whole operator against grid copying, ms of GPU kernel time (all kernels), one row
of $n$ float32 elements; same profiles as Table~\ref{tab:materialized}.}
\label{tab:whole}
\small
\begin{tabular}{@{}lrrrrrrrr@{}}
\toprule
& \multicolumn{2}{c}{T4} & \multicolumn{2}{c}{L4} & \multicolumn{2}{c}{A100} & \multicolumn{2}{c}{H100} \\
data, $n$ & seg. & grid & seg. & grid & seg. & grid & seg. & grid \\
\midrule
increasing, $10^7$ & 3.25 & 29.4 & 2.01 & 23.8 & 1.33 & 5.68 & 0.83 & 3.18 \\
increasing, $10^8$ & 42.8 & 396 & 19.4 & 345 & 11.9 & 63.6 & 7.28 & 34.9 \\
cliff, $10^8$ & 41.6 & 386 & 17.9 & 340 & 10.4 & 61.5 & 6.39 & 33.8 \\
noise, $10^7$ & 2.01 & 2.22 & 1.10 & 1.21 & 0.75 & 0.84 & 0.49 & 0.54 \\
noise, $10^8$ & 18.2 & 23.3 & 7.99 & 10.4 & 4.25 & 5.88 & 2.67 & 3.65 \\
\bottomrule
\end{tabular}
\end{table*}

Table~\ref{tab:whole} gives the whole operator. Only the chunk stage is the same kernel in
both. The control reconstructs the fit from its flat array with a kernel of its own, and before
its first merge it tags every vertex slot with the vertex's chunk, since a copied vertex leaves
its chunk's slot. On dense hulls the merge stage accounts for most of the difference: the
operator with segment lists is $9.1$ to $9.3$ times faster on the T4, $12$ to $19$ times on the
L4, $4.3$ to $5.9$ times on the A100 and $3.9$ to $5.3$ times on the H100. On noise at $10^7$
elements the merge stages take the same time, and the operator's difference of 10 to 13\%
comes from the control's other kernels: on the T4 its reconstruction takes $0.23$ against
$0.19$~ms and its tagging and buffers $0.28$ against $0.10$~ms. Part of this is the cost of a
flat layout and part that of the control's implementation, so we do not count it for the
segment lists. At $10^8$ noisy elements the merge stage contributes $1.6$ of the T4's
$5.1$~ms difference.

\subsection{Ablations}
\label{sec:ablations}
Each variant changes one choice of the default. Table~\ref{tab:ablations} gives its time
relative to the default's on rows of $10^8$ elements, for the chunk stage in the first variant
and the merge stage in the others.
\begin{table}[t]
\caption{Ablations: time of a variant over the default's, one row of $10^8$ float32
elements; $>1$ means the default is faster.}
\label{tab:ablations}
\small
\setlength{\tabcolsep}{3.5pt}
\begin{tabular}{@{}llrrrr@{}}
\toprule
variant (stage) & data & T4 & L4 & A100 & H100 \\
\midrule
stack in local memory & noise & 2.10 & 1.58 & 1.77 & 1.81 \\
\quad instead of a bit mask (chunk) & increasing & 3.37 & 1.25 & 1.05 & 1.17 \\
three-case search, cap & noise & 1.89 & 2.89 & 1.91 & 1.99 \\
\quad growing with the level (merge) & cliff & 1.05 & 1.08 & 1.03 & 1.03 \\
turns on absolute & noise & 1.02 & 0.93 & 0.90 & 0.88 \\
\quad coordinates (merge) & sawtooth & 0.84 & 0.76 & 0.80 & 0.82 \\
a warp for every & noise & 2.37 & 3.43 & 1.37 & 1.29 \\
\quad merge (merge) & sawtooth & 0.72 & 0.57 & 0.39 & 0.36 \\
warps for capped & noise & 1.15 & 0.95 & 1.10 & 1.06 \\
\quad walks only (merge) & sawtooth & 0.71 & 0.52 & 0.60 & 0.55 \\
\bottomrule
\end{tabular}
\end{table}
The bit-mask lane hulls are faster than a stack in local memory with a pop loop on every
GPU, by $1.6$ to $2.1$ times on noise. On increasing data, where the stack is deepest, the
difference is $3.4$ times on the T4, which has the least local-memory bandwidth, and $1.05$ to
$1.25$ times on the others. The three-case search of Theorem~\ref{thm:merge}, with a walk cap
that grows with the level, is two to three times slower on noisy data and about as fast on
cliffs, where bridges are long. Plain double predicates on absolute coordinates, without the
coordinate hierarchy and its filter, are up to 24\% faster on sawtooth data, where merges walk
far, and up to 12\% on noise; the accuracy of Section~\ref{sec:numerics} depends on the
hierarchy. The last two variants change the schedule. A warp for every merge makes the merge
levels $1.4$ to $2.8$ times faster on sawtooth data, where many walks reach their cap and a
single thread then searches, and $1.3$ to $3.4$ times slower on noise, where many short merges
leave most of a warp idle. In the last variant a thread walks first and a second kernel
finishes the merges whose walk reached its cap, with a warp each. It makes sawtooth merges $1.4$ to $1.9$ times faster on all four GPUs, but is up to 15\%
slower on noise on three of them, which is why the default keeps the single rule.

\subsection{Performance breakdown}
\label{sec:breakdown}
On one row of $10^8$ float32 elements the fit takes $20.0$, $8.4$, $4.3$ and $2.7$~ms on the
T4, L4, A100 and H100 (medians). A device-to-device copy of the input takes $3.4$, $3.5$, $0.6$
and $0.27$~ms, so the fit takes $5.9$, $2.4$, $7$ and $10$ times as long as the copy. On noise the
chunk stage takes $63$ to $73\%$ of the kernel time on the four GPUs, the merge levels $10$ to
$15\%$, the reconstruction $7$ to $19\%$ and the prefix sums the rest. The chunk stage is limited
by the instructions of the lane hulls and their merges. The fit is therefore
furthest from the copy on GPUs with the most bandwidth per unit of compute, and closest on the
L4, whose memory is slowest relative to its arithmetic. With every point a vertex (increasing
data) the same row takes $42.8$, $19.4$, $11.9$ and $7.3$~ms of kernel time, $2.4$ to $2.8$
times as long as noise, since the chunk stage writes every vertex and the fit reads them all.
On the T4 the reconstruction's share grows from $10$ to $42\%$. On shorter rows the shares
change. At $10^7$ noisy elements the merge levels take $25$, $37$, $45$ and $48\%$ of the kernel
time on the four GPUs, and on the H100 as long as the chunk stage ($0.23$ against $0.20$~ms).
Each of the $14$ levels costs about $0.016$~ms however few merges it has, which is the latency
of a launch.

\subsection{Numerical accuracy}
\label{sec:numerics}
The reference is a float64 solution on the CPU of the same cast input, so that the comparison
measures the solver and not the cast; it is the library's own float64 pool adjacent violators,
which the test suite checks against SciPy. Ten data families (increasing, decreasing, constant,
piecewise constant with noise, Gaussian, random walk, sawtooth, a cliff, nearly collinear
points, and data offset by $10^6$) are crossed with $n \in \{10^3, 10^5, 10^7\}$, both
precisions, and three weightings: none, uniform in $[0.1, 1.1]$, and $e^{U[0, 8]}$, a dynamic
range of $3 \cdot 10^3$. We report four quantities:
\begin{itemize}
\item the largest step down of the output, relative to $\max |y|$;
\item the largest deviation from the reference, $\max_i |z_i - z^\star_i| / (1 + |z^\star_i|)$;
\item the excess of the weighted least-squares objective over the reference's, relative to the
reference's, or to $\varepsilon \sum_i w_i y_i^2$ where that is near zero;
\item the largest violation of the KKT conditions, computed in float64 from the output alone,
relative to $2 \sum_i w_i |y_i|$. The multipliers $\lambda_i = -2 \sum_{j \le i} w_j (z_j - y_j)$
must be non-negative, zero at the end, and zero where the output steps up. This measure
shrinks with $n$ for a fixed error per element, which is why the deviation is reported too.
\end{itemize}
We also compare the number of blocks with the reference's.

In float64 no output steps down on these families. The deviation is at most
$1.3 \cdot 10^{-10}$ (extreme weights, $n = 10^7$), the objective is within $7 \cdot 10^{-10}$,
the KKT residual is at most $8 \cdot 10^{-14}$, and the number of blocks is the same
everywhere. This includes $10^7$ increasing points with extreme weights, whose cumulative sums
reach $10^{10}$. A version that forms slopes from differences of absolute sums, without the
coordinate hierarchy, pools $681$ of these points that are vertices of the reference and steps
down by up to $9 \cdot 10^{-7}$.

In float32, over all families, sizes and weightings, the deviation is at most
$1.2 \cdot 10^{-7}$, the objective within $7 \cdot 10^{-9}$ and the KKT residual below
$6 \cdot 10^{-8}$. No output steps down, here or in the test suite, which checks monotonicity
on its own families with three weightings. On $10^7$ consecutive floats, each one ulp above the
previous, every point stays its own block without weights, as in the reference. With weights,
about 8\% of neighbouring outputs coincide: the products $w_i y_i$ are rounded to float32
before the hull is built, and neighbouring slopes may then round to the same value.

A further family stresses cancellation. Each 64 elements hold 32 values of $+B$, one of $-32B$
and 31 small, nearly collinear values, for $B$ from $10^3$ to $10^{12}$ and $n$ from 500 to
$10^5$. The cumulative sums rise and fall back within every chunk, so a coordinate is the
difference of terms far larger than itself. In float32 the output does not step down and
deviates from the reference by at most $4 \cdot 10^{-9}$ of $\max |y|$. In float64 it deviates
by at most $1.4 \cdot 10^{-15}$ of $\max |y|$, and neighbouring outputs step down by up to
$5 \cdot 10^{-16}$ of $\max |y|$, a few units in the last place of the data's scale. Monotone
output is therefore not guaranteed in float64 either. A cumulative maximum over the output
would enforce it at the cost of one more pass over the data; the operator does not apply one.

These results use the default settings. The option \texttt{exact=False} drops the float32
pairs within the lanes and the filter of the chunk level; float64 is unaffected, as its sums use
no pairs. On a T4 it takes $1.30$ against $2.30$~ms for $10^4$ rows of $10^3$ elements and
$13.7$ against $20.0$~ms for one row of $10^8$ (medians). Without the pairs, a lane's
cumulative sum of 32 terms can be 32 times a single term, or $32 e^8$ times with extreme
weights. Its float32 rounding then loses the last bits of each increment and decides which
points are vertices. On the $10^7$ unweighted consecutive floats, \texttt{exact=False} keeps
$1.17$ million of the $10^7$ blocks. It pools 88\% of the points and so has less to merge and
fit, which is part of its speed. With extreme weights it deviates from the reference by up to
$4.4 \cdot 10^{-3}$ and steps down by up to $8.5 \cdot 10^{-4}$ of the data's scale. 

\subsection{Case study: calibrating a click model}
\label{sec:case}
Isotonic regression is a standard non-parametric calibration of classifier
scores~\cite{zadrozny2002,niculescu2005}. We reproduce a common industrial setting on the
Criteo display-advertising data~\cite{criteo2014}, 45.8 million impressions over seven days in
time order. A logistic regression with hashed features is trained on the first five days with
a tenth of the non-clicks, a negative downsampling that leaves the predicted click
probabilities far too high~\cite{he2014}. The sixth day calibrates the model and the seventh
evaluates it, with one curve for all traffic and, separately, one curve for each of the 280
frequent values of a categorical feature. Isotonic calibration brings the expected calibration
error from $0.41$ to $0.0015$; Appendix~\ref{app:calibration} compares it with the analytic
correction and Platt scaling.

The fits agree with scikit-learn's to $6 \cdot 10^{-16}$ for one curve and to
$7 \cdot 10^{-10}$ per segment, where the keys (segment plus score) are float64. We time the
fits alone on both sides. On a T4, one curve over the $6.5$ million scores of a day takes
$0.042$~s, and scikit-learn takes $2.0$~s on one core of the host CPU. Over all $45.8$ million
scores the times are $0.155$ and $17.2$~s, including the sort of the scores and the pooling of
equal ones. With the scores in host memory, as scikit-learn's are, the copy to the GPU doubles
the time to $0.31$~s. The 280 per-segment curves take $0.033$ against $0.77$~s. They are fitted
as one problem: the scores are keyed by segment and the targets of each segment offset by twice
its index, so that no block can extend across a segment boundary. The operator takes a
rectangular batch, so segments of unequal length are either padded, exactly, with values above
the row's maximum, or laid end to end in one row as here; the merge tree then spreads that row
across the whole GPU.

\subsection{Gradients}
\label{sec:gradients}
Table~\ref{tab:gradients} gives the fit alone and with its gradient with respect to $y$, for a
random $g$. On rows of $10^3$ elements the backward pass adds a quarter to three quarters of the
fit's time. On a row of $10^7$ elements and on $10^6$ rows of 16 it adds one to two times the
fit's time; on the short rows it pays a fixed cost per tile, while the forward pass packs several
rows to a warp. The lead over torchsort is therefore smaller with the gradient than for the fit
alone.

\begin{table*}[t]
\caption{Fit and fit plus gradient, ms, decreasing unweighted float32; median of ten calls timed
with CUDA events. The interquartile range is within 4\% of the median except for ours with the
gradient on $10^6$ rows of 16 on the T4 (17\%).}
\label{tab:gradients}
\small
\begin{tabular}{@{}lrrrrrrrr@{}}
\toprule
& \multicolumn{2}{c}{T4} & \multicolumn{2}{c}{L4} & \multicolumn{2}{c}{A100} & \multicolumn{2}{c}{H100} \\
batch $\times$ $n$ & fit & +grad & fit & +grad & fit & +grad & fit & +grad \\
\midrule
$10^6 \times 16$ & 2.17 & 4.78 & 0.72 & 2.06 & 0.51 & 1.53 & 0.29 & 0.82 \\
\quad torchsort & 74.4 & 85.2 & 7.03 & 8.95 & 3.43 & 4.43 & 1.80 & 2.33 \\
$10^4 \times 10^3$ & 2.27 & 2.86 & 0.65 & 1.12 & 0.56 & 0.72 & 0.34 & 0.43 \\
\quad torchsort & 99.4 & 106 & 14.0 & 17.6 & 19.0 & 23.5 & 12.6 & 15.6 \\
$1 \times 10^7$ & 2.10 & 4.17 & 1.23 & 3.02 & 0.89 & 2.30 & 0.57 & 1.75 \\
\quad torchsort & 5\,491 & 6\,898 & 4\,948 & 6\,376 & 6\,423 & 8\,157 & 4\,724 & 5\,961 \\
\bottomrule
\end{tabular}
\end{table*}

\subsection{Differentiable ranking and a training step}
\label{sec:e2e}
The soft rank of Blondel et al.~\cite{blondel2020} with $L_2$ regularization of strength
$\varepsilon$ projects $\theta = x/\varepsilon$ onto the permutahedron: it sorts $\theta$ in
decreasing order into $s$, solves a decreasing isotonic regression of $z = s - w$ with
$w = (n, \dots, 1)$, and returns $s$ minus that fit in the original order. We time torchsort's
\texttt{soft\_rank} against a copy of its implementation in which only the isotonic solver,
forward and backward, is replaced by \sysname{}; the sort, the inverse permutation and the
backward pass around the solver are torchsort's code. The values are standard normal
(Table~\ref{tab:e2e}). We use two strengths: $\varepsilon = 1$, torchsort's
default, at which the ranks of these values are pooled into one block per row, and
$\varepsilon = 1/n$, at which about 70\% of the elements start a block and the soft ranks are
close to the hard ones.

\begin{table*}[t]
\caption{Soft rank (forward and backward) and a training step with a Spearman loss, ms, median of
20 calls, $\varepsilon = 1/n$, torchsort / ours. The two differ only in the isotonic solver. The
training step is a forward, backward and Adam step of a two-layer scoring model on $n$ items per
query.}
\label{tab:e2e}
\small
\begin{tabular}{@{}llrrrrrrrr@{}}
\toprule
& & \multicolumn{2}{c}{T4} & \multicolumn{2}{c}{L4} & \multicolumn{2}{c}{A100} & \multicolumn{2}{c}{H100} \\
& batch $\times$ $n$ & torchsort & ours & torchsort & ours & torchsort & ours & torchsort & ours \\
\midrule
soft rank & $1024 \times 100$ & 2.42 & 0.97 & 1.94 & 1.01 & 2.63 & 1.32 & 1.53 & 0.50 \\
& $256 \times 10^3$ & 6.77 & 0.96 & 4.34 & 1.06 & 6.02 & 1.42 & 4.26 & 0.52 \\
& $64 \times 10^4$ & 41.9 & 1.60 & 25.0 & 1.62 & 35.7 & 2.13 & 27.5 & 0.95 \\
& $8 \times 10^5$ & 208 & 2.24 & 147 & 1.59 & 194 & 1.70 & 160 & 0.93 \\
& $1 \times 10^6$ & 879 & 3.74 & 638 & 1.55 & 905 & 1.68 & 718 & 0.88 \\
\midrule
training step & $256 \times 10^3$ & 14.4 & 7.98 & 11.4 & 6.78 & 8.66 & 2.93 & 5.37 & 1.72 \\
& $32 \times 10^4$ & 46.0 & 10.3 & 32.4 & 8.62 & 38.5 & 3.68 & 29.0 & 2.14 \\
& $4 \times 10^5$ & 124 & 13.2 & 92.3 & 10.8 & 121 & 4.22 & 89.5 & 2.27 \\
\bottomrule
\end{tabular}
\end{table*}

With $\varepsilon = 1/n$, forward and backward together are $1.9$ to $820$ times faster with
\sysname{} on the four GPUs, the gain growing with the length of the rows; with $\varepsilon = 1$ they are $1.8$ to $382$ times faster. The training step is a synthetic
benchmark of a compact scoring model: a two-layer network ($32 \to 128 \to 1$) scores each item,
the soft rank of the scores is taken per query, and the loss is one minus the Spearman correlation
with the target ranks. Its gain depends on the soft rank's share of the step. For $4$ queries of
$10^5$ items with $\varepsilon = 1/n$, the soft rank takes nearly all of torchsort's step on every
GPU and 12 to 16\% of ours on the T4 and L4, and the step is $8.5$ to $9.4$ times faster there. On
the A100 and H100 the network is faster, the soft rank takes 43 to 44\% of our step, and the step
is $29$ to $40$ times faster. For $256$ queries of $10^3$ items the step is $1.7$ to $3.1$ times
faster. With $\varepsilon = 1$ the steps are $1.5$ to $17$ times faster. The model is cheap, so the
soft rank's share is large; with a heavier model the gain shrinks with that share.

We checked both solvers against a float64 solution of the same float32 problem $z$. The
values and gradients of \sysname{} agree with it closely on all four GPUs, except on two of them
at one boundary between two blocks of exactly equal means, where the fit is not
differentiable. torchsort's values deviate
by up to $2.5 \cdot 10^{-4}$ with $\varepsilon = 1$, where it accumulates long blocks in
float32, and its gradients by up to $4.6\%$ in the $\ell_2$ norm with $\varepsilon = 1/n$ and
$n \ge 10^4$, on all four GPUs. Appendix~\ref{app:grad} gives the measures and their analysis.

\section{Related Work}
\label{sec:related}
\paragraph{Isotonic regression.}
The equivalence of isotonic regression with the greatest convex minorant of the cumulative
sum diagram is classical~\cite{barlow1972,robertson1988}. Stout~\cite{stout2022} surveys the
sequential algorithms, from pool adjacent violators to the $O(n)$ bounds for the $L_2$ case.
Kearsley, Tapia and Trosset~\cite{kearsley1996} and Kearsley~\cite{kearsley2006} parallelize by partitioning the
row, solving the parts and merging their solutions, the scheme behind Spark's implementation.
The merge of the parts' solutions is sequential there; segment lists are a way to do this merge
on a GPU. Henzi, M\"osching and D\"umbgen~\cite{henzi2022} accelerate the algorithm for
isotonic distributional regression, where one data set yields a batch of weighted problems,
the setting of our operator. On accelerators, Sander et al.~\cite{sander2023} solve isotonic
regression approximately by Dykstra's projections inside a differentiable top-$k$.
torchsort~\cite{torchsort}, built for the soft sorting and ranking of Blondel et
al.~\cite{blondel2020}, runs the exact pool adjacent violators with one thread per row,
forward and backward, and is our main baseline.

\paragraph{Parallel hulls of sorted points.}
The merge of two $x$-separated hulls along a bridge is due to Preparata and
Hong~\cite{preparata1977}, and its $O(\log n)$ search to Overmars and van
Leeuwen~\cite{overmars1981}. Their dynamic hull keeps chains in concatenable queues, balanced
trees that are split at a bridge and concatenated in $O(\log n)$, so a merge relinks nodes
instead of copying vertices. Segment lists apply the same idea in a flat form: at most one
descriptor per chunk, cut at the bridge, concatenated by moving descriptors, and read by
indexing. G{\ae}de et al.~\cite{gaede2024} reduced the case analysis of Overmars and van
Leeuwen to three cases; we use their form, mirrored to lower hulls, in
Theorem~\ref{thm:merge}. For presorted points, Goodrich~\cite{goodrich1987} reaches $O(\log n)$
depth on a CREW PRAM. Berkman, Schieber and Vishkin~\cite{berkman1996} reach $O(\log\log n)$ on
a CRCW PRAM with an implicit representation of the hull, since writing it into a contiguous
array, a compaction, would take $\Omega(\log n / \log\log n)$ time. Our layout, one slot per
chunk, serves the same purpose on a GPU. Browne et al.~\cite{browne2023} give an
$O(n)$-work, $O(\log n)$-span algorithm in the binary forking model.

\paragraph{Hulls on GPUs.}
Work on two-dimensional hulls on GPUs treats unsorted point clouds, with QuickHull and
parallel partitioning~\cite{srungarapu2011,tzeng2012} or by filtering most points with extreme
witnesses before building the hull of the rest, as Keith, Ferrada and Navarro
do~\cite{keith2022}. Presorted input needs no such filter. Closest in design is the
warp-cooperative three-dimensional hull in CuACD~\cite{cuacd2026}. A warp there partitions
its points among lane groups, merges the sub-hulls pairwise in logarithmically many rounds,
and finds bridges by gift-wrapping with warp shuffles. We also use the warp as the unit of
design and merge in pairwise rounds. The differences are the input (sorted, two-dimensional,
up to $10^8$ points in one problem), the per-lane bit-mask stack, the walk whose cost is
amortized against removed vertices, the hierarchy with segment references, and floating-point
predicates with error bounds in place of integer determinants. Stream
compaction~\cite{billeter2009} is the usual GPU method to keep the survivors of a step; the
grid-copying control of Section~\ref{sec:vs-materialized} uses it. Choosing a thread, a warp
or more per unit of work by its size is common, as in the CSR-Adaptive sparse matrix-vector
product~\cite{greathouse2014}, and our per-level rule is of this kind. Our hi--lo pairs use
float-float arithmetic on GPUs~\cite{dagraca2006}. Floating-point filters for robust
predicates go back to Shewchuk~\cite{shewchuk1997}, and Bartels, Fisikopoulos and
Weiser~\cite{bartels2023} derive static and semi-static filters. Our filter compares its
bound by exponents only and falls back to double arithmetic on local coordinates
(Appendix~\ref{app:filter}).

\section{Conclusion}
\label{sec:conclusion}
The fit of isotonic regression can be computed in parallel within a row on a GPU by merging
lower hulls. Segment lists lower the cost of data movement in that merging: they remove the
copying of surviving vertices, which dominates on dense hulls and matters little on sparse
ones. The resulting operator is fast enough to change the cost of a training step that uses a
differentiable soft rank.

\paragraph{Limitations.}
The operator takes a rectangular batch; rows of unequal length are padded or laid end to end.
The merge levels follow a fixed schedule, which is far from the best on some data, and on
mid-sized rows they cost more in launches than in work. The executed code adds sequential rounds of
descriptor copying to the depth of the model (Proposition~\ref{prop:merge-code}). The predicates are filtered and
their second evaluation is double arithmetic, so accuracy and monotone output are measured
properties, and near ties the gradient follows the solver's partition.

\bibliographystyle{ACM-Reference-Format}
\bibliography{refs}

\appendix

\section{The three-case bridge search}
\label{app:three}
The nested search of Section~\ref{sec:analysis-merge} tests a candidate vertex of one chain
against its tangent to the other, a bisection inside a bisection. The search below halves a
candidate range on one of the two chains at every step instead.

\begin{lemma}[three cases, lower hulls]
\label{lem:threecases}
Let $A$ and $B$ be the hulls of two neighbouring nodes without their shared point $s$, so
that $x(a) < x(s) < x(b)$ for $a \in A$, $b \in B$, and suppose $s$ is not a vertex of
$H(A \cup B)$. Let $(a^*, b^*)$ be the bridge, with $a^*$ the leftmost and $b^*$ the rightmost
point of contact of its line. For an edge $\alpha$ of $A$, an edge $\beta$ of $B$, and points
$l \in \alpha$, $r \in \beta$:
(1) if $\mathrm{slope}(\alpha) \ge \mathrm{slope}(lr)$ then $a^*$ is not right of
$\alpha$'s left end point;
(2) if $\mathrm{slope}(lr) \ge \mathrm{slope}(\beta)$ then $b^*$ is not left of
$\beta$'s right end point;
(3) otherwise the lines of $\alpha$ and $\beta$ meet at a point $\gamma$ with
$x(l) < x(\gamma) < x(r)$, and if $x(\gamma) \le x(s)$ then $a^*$ is not left of $\alpha$'s
right end point, while if $x(\gamma) \ge x(s)$ then $b^*$ is not right of $\beta$'s left end
point.
\end{lemma}
\begin{proof}
Let $h$ be the convex function whose graph is $H(A \cup B)$ and $\ell$ the line of the bridge,
which supports $A$ and $B$ from below.
(1) Suppose $a^*$ lies right of $\alpha$'s left end point; as $a^*$ is a vertex, $\alpha$ is
then an edge of $h$ left of $a^*$, so the line of $\alpha$ supports $h$ and $r_y \ge
h(x(r)) \ge l_y + \mathrm{slope}(\alpha)(x(r) - x(l))$, i.e.\ $\mathrm{slope}(lr) \ge
\mathrm{slope}(\alpha)$. With the hypothesis this is an equality, $r$ lies on the line of
$\alpha$, $h$ is linear on $[x(l), x(r)] \ni x(a^*)$, and $a^*$ is not the leftmost point of
contact, a contradiction. (2) is symmetric.
(3) The slopes differ, so the lines meet. From $\mathrm{slope}(lr) > \mathrm{slope}(\alpha)$
the point $r$ lies strictly above the line of $\alpha$ and on that of $\beta$, which is below
the line of $\alpha$ left of $\gamma$; so $x(r) > x(\gamma)$, and symmetrically $x(l) <
x(\gamma)$. Let $x(\gamma) \le x(s)$ and suppose $a^*$ is not right of $\alpha$'s left end
point. Then $\alpha$ is not before the edge out of $a^*$ in $A$, so $\mathrm{slope}(\ell) \le
\mathrm{slope}(\alpha)$, and $\ell$ is not above the line of $\alpha$ at $\alpha$'s left end
point, hence nowhere to its right. At $x(b^*) > x(s) \ge x(\gamma)$ the line of $\alpha$ is
strictly below the line of $\beta$, which is not above $B$: $b^*_y = \ell(x(b^*)) <
\beta\text{-line}(x(b^*)) \le b^*_y$, a contradiction. The case $x(\gamma) \ge x(s)$ is
symmetric.
\end{proof}

Each case discards a prefix or suffix of one chain, so the search keeps a contiguous
candidate range on each side and halves one of them per step. On segment lists the halving
is done in two levels: while a range spans several segments, $\alpha$ is the edge between
the middle segment's last vertex and the next segment's first (an edge of the chain, as
consecutive vertices of the hull), and a case keeps either the segments up to the middle one
or those after it; within one segment the vertices are consecutive slots and $\alpha$ is the
middle edge. A side of $k$ segments is thus reduced to a vertex in $\log_2 k + \log_2 C +
O(1)$ steps, each reading two vertices and two descriptors, and the remaining tangent from
that vertex to the other side is the same bisection with a turn predicate. The predicates are
signs of polynomials of degree two (slope comparisons) and three ($x(\gamma)$ against
$x(s)$) in the coordinates; any interior points $l$, $r$ serve, so the edge midpoints are
used.

\begin{theorem}[merge phase, three-case search]
\label{thm:merge}
With every walk capped at a constant $T$ steps and the search above beyond it, and with a
parallel copy for the descriptor moves, the merge phase over $m = \lceil n/C \rceil$ chunk
hulls does $O(n + m \log m)$ work at depth $O(T \log n + \log^2 n)$; the work is $O(n)$
whenever $C \ge \log_2 m$.
\end{theorem}
\begin{proof}
The walks and the moves cost as in Theorem~\ref{thm:merge-impl}. A search at level $r$ takes
at most $2(r + \log_2 C) + O(1)$ steps, so the searches cost
$\sum_r (m / 2^{r+1} + 1)\, O(r + \log C) = O(m \log C + \log^2 m) = O(n)$. The moves cost
$O(m \log m)$, which is $O(n)$ when $C \ge \log_2 m$. The depth of level $r$ is at most
$T + 2(r + \log_2 C) + O(1)$, and the sum over $r \le \log_2 m$ is
$O(T \log m + \log^2 m + \log m \log C) = O(T \log n + \log^2 n)$.
\end{proof}

Theorem~\ref{thm:merge} bounds a variant of the code. We implemented the three-case search
for the levels run with a thread per merge, with a walk capped at $8(r + \log_2 C)$ steps. The
levels run with a warp per merge keep the warp-parallel nested search, so the variant does not
reach the bound at those levels, and its times show the cost of the search where it is used. On
an A100 the variant and the default take the same time when bridges reach across the whole row,
with a last point far below the rest: $1.41$ against $1.37$~ms for the merges of $10^8$
elements. On noisy data the variant's merges are twice as slow ($1.16$ against $0.58$~ms): few
walks reach the cap there, and the variant's larger kernel lowers occupancy.

\section{The filters of the predicates}
\label{app:filter}
\paragraph{Merge levels.}
The merge levels decide turns $\mathrm{turn}(a, b, d) = (b_x - a_x)(d_y - a_y) - (b_y -
a_y)(d_x - a_x)$ of vertices addressed by segment and slot. A coordinate of a vertex of chunk
$j$ is represented as $X = c_j + c'_j + h$. Here $c_j$ and $c'_j$ are the hi and lo parts of
the chunk's start, a renormalized double-double pair with $|c'_j| \le u |c_j|$, $h$ is the
vertex's coordinate relative to the chunk, a double, and $u = 2^{-53}$. The slow path computes
the turn from differences of these parts (Section~\ref{sec:implementation}). The fast path
uses absolute coordinates $\hat X = \mathrm{fl}(c_j + h)$, one addition each. It computes
$b'_x = \mathrm{fl}(\hat b_x - \hat a_x)$ and likewise $b'_y, d'_x, d'_y$, then
$p_1 = \mathrm{fl}(b'_x d'_y)$, $p_2 = \mathrm{fl}(b'_y d'_x)$ and $c = \mathrm{fl}(p_1 - p_2)$.

\paragraph{Error.}
Let $M$ be at least the magnitude of each of the six absolute coordinates and of each chunk
start $c_j$ of the three points. Then $|\hat X - X| \le |c'_j| + u |\hat X| \le 2uM$ to first
order. The chunk starts matter for the ordinates: a cumulative sum can be far from zero at a
chunk's start and return close to zero within the chunk, so that $|c_j|$ exceeds the
coordinate. For the abscissae, sums of positive weights, the chunk start never exceeds the
coordinate. With $\Sigma = |b'_x| + |b'_y| + |d'_x| + |d'_y|$, the errors of the differences
are at most $4uM + u|b'_x|$ and so on. Those of the products are at most
$4uM(|b'_x| + |d'_y|) + 3u|p_1|$ and $4uM(|b'_y| + |d'_x|) + 3u|p_2|$, and that of $c$ is at
most
\begin{equation}
  4uM\Sigma + 3u(|p_1| + |p_2|) + u|c| \le 4u\,(M\Sigma + |p_1| + |p_2|)
  \label{eq:filter-error}
\end{equation}
to first order, as $|c| \le (1 + u)(|p_1| + |p_2|)$. The code uses twice this,
$E = 8u\,(M\Sigma + |p_1| + |p_2|)$, which absorbs the second-order terms. If $|c| > E$, then $c$
has the sign of the turn of the represented points.

\paragraph{The test by exponents.}
For a normal $v \neq 0$ let $e(v) = \lfloor \log_2 |v| \rfloor$, so that $2^{e(v)} \le |v| <
2^{e(v) + 1}$. Let $e_a$ be the larger exponent of $\hat a$'s two coordinates, $e_s$ the largest
exponent of the ordinates' chunk starts of the three points, $e_d$ the largest exponent of the
four differences, and $e_m = \max(e_a, e_s, e_d) + 1$. Every coordinate of $\hat b$ and
$\hat d$ is at most $|\hat a| + |\text{difference}| < 2^{e_a + 1} + 2^{e_d + 1} \le
2^{e_m + 1}$, and every chunk start is below $2^{e_s + 1} \le 2^{e_m}$, so $M = 2^{e_m + 1}$
qualifies. Further, $\Sigma < 4 \cdot 2^{e_d + 1} = 2^{e_d + 3}$ and
$|p_1| + |p_2| \le \max(|d'_x|, |d'_y|)(|b'_x| + |b'_y|) < 2^{e_d + 1}\Sigma \le 2^{e_m}\Sigma$.
Hence
\begin{equation}
  E < 8u\,(2^{e_m + 1} + 2^{e_m})\,\Sigma < 2^{-50} \cdot 2^{e_m + 2} \cdot 2^{e_d + 3}
  = 2^{e_m + e_d - 45}.
\end{equation}
Since $|c| \ge 2^{e(c)}$, the condition
\begin{equation}
  e(c) \ge e_m + e_d - 44
  \label{eq:filter-test}
\end{equation}
implies $|c| \ge 2^{e_m + e_d - 44} > E$, and the fast result is returned. The code reads the
biased exponents $\tilde e = e + 1023$ from the high word of each double with an integer shift
and mask. In these, \eqref{eq:filter-test} reads $\tilde e(c) \ge \tilde e_m + \tilde e_d - 1067$
with $\tilde e_m = \max(\tilde e_a, \tilde e_s, \tilde e_d) + 1$, which takes integer maxima and
one comparison. The derivation assumes normal numbers, so a zero or subnormal $c$ or
difference, whose biased exponent is $0$, always takes the slow path. Replacing each quantity
by a power of two makes the test coarser than $|c| > E$ by a constant factor, which sends a few
more nearly collinear triples to the slow path.

\paragraph{What it certifies.}
A turn that passes the test has the sign that exact arithmetic gives on the represented points,
the inputs of the slow path. A turn that fails it is computed by the slow path, in double
arithmetic on differences at the scale of the chunks' sums, which can still err for nearly
collinear points. The represented points are themselves not the exact cumulative sums: they
carry the rounding of the chunks' sums and, in float32, of the products $w_i y_i$. The
predicates are therefore not exact, and the accuracy of Section~\ref{sec:numerics} is a
measurement.

\paragraph{Lane and chunk levels.}
Here the coordinates are pairs of the working type, with unit roundoff $u$, and the fast path
uses their hi parts. Within a lane, a coordinate is a renormalized pair whose hi part is its
value rounded once, so it is within $u$ times the largest magnitude of the lane's sums, the
bound used by the lane chain. Across the lanes of a chunk, a coordinate is the sum of a lane's
start and a sum within the lane, both pairs, and its hi part is the rounded sum of their hi
parts. It is within $u(|y_h| + |s| + |t|) \le 2u\,m$ of the pair's value, where $s$ and $t$ are
the two terms and $m = |s| + |t|$; $m$ can far exceed $|y_h|$ when the terms cancel, and the
code carries it with each point. For the abscissae, sums of positive weights, $|x_h|$ bounds
the terms. With $M_y$ the largest $m$ of the three points and $M_x$ their largest $|x_h|$, the
cross product is within $4u\,(M_y(|x_1| + |x_2|) + M_x(|y_1| + |y_2|)) + 4u\,(|p_1| + |p_2|)$ to
first order. The code uses $12u$ and $8u$ for the two coefficients. The bound is evaluated in the
working type, as these levels run at the single-precision rate, and a turn within it is
recomputed from the pairs with two-sum and two-product.

\section{Gradient validation}
\label{app:grad}
\paragraph{Tests.}
The test suite checks the gradients against finite differences on the CPU, in $y$ and in $w$,
in float64. On the GPU it checks them against the CPU's on data without near-ties, and against
block sums over the GPU's own block starts on every data family, sawtooth and constant data
included. Separately, it checks blocks whose means differ but round to the same float32 value:
they stay separate blocks.

\paragraph{Soft rank.}
For the soft rank of Section~\ref{sec:e2e} we compare the solvers with a float64 solution of
the same float32 problem $z = s - w$, on each of the four GPUs; the inputs are drawn on the GPU,
so they differ between GPUs. The values of \sysname{} agree with it to $5 \cdot 10^{-8}$ of their
scale. Those of torchsort agree to $6 \cdot 10^{-8}$ with $\varepsilon = 1/n$; with
$\varepsilon = 1$, where each row is one long block whose mean torchsort accumulates in float32,
they deviate by $6 \cdot 10^{-7}$ at $n = 100$, by up to $7 \cdot 10^{-6}$ for $n$ up to $10^5$,
and by $2.5 \cdot 10^{-4}$ at $n = 10^6$. For the gradients we compare the vector-Jacobian product of the isotonic step, $A g$
for a standard normal $g$, where $A$ averages over the blocks. We report its largest error
relative to the reference's largest entry ($\ell_\infty$; the reference's entries reach $4.3$ to
$5.0$ with $\varepsilon = 1/n$) and its error in the $\ell_2$ norm relative to the reference's.
Since $A g$ is constant on each block and differs between blocks, its changes mark the blocks each
solver used.

The gradients of \sysname{} agree with the reference to $5.2 \cdot 10^{-7}$ in $\ell_\infty$,
with the same blocks, except on the row of $10^6$ with $\varepsilon = 1/n$ on the T4 and the A100.
There one of the reference's $712\,735$ and $712\,368$ block boundaries differs, and the means on
its two sides are exactly equal: the fit is not differentiable at that point, and each partition
gives a valid one-sided derivative. That boundary accounts for an $\ell_\infty$ difference of
$0.21$ and $0.19$ and an $\ell_2$ difference of $1.4 \cdot 10^{-3}$ and $1.2 \cdot 10^{-3}$.
torchsort's gradients agree to $1.3 \cdot 10^{-5}$ with $\varepsilon = 1$ and for $n \le 10^3$.
With $\varepsilon = 1/n$ and $n \ge 10^4$, 7 to 1\,425 of its block boundaries differ from the
reference's. At each of them the means on the two sides differ, by up to $5 \cdot 10^{-8}$ of the
data's scale, which is below float32's resolution. These differences are therefore not explained
by equal means of neighbouring blocks; we compared only those means, not the stability of the
partition within the blocks. torchsort's gradient then differs by $0.42$ to $0.76$ in $\ell_\infty$
and by $0.5$ to $4.6\%$ in $\ell_2$.

\section{Calibration metrics}
\label{app:calibration}
\begin{table}[h]
\caption{Calibration of the click model of Section~\ref{sec:case} on day 7: expected
calibration error (ECE) and the change in log-loss, with 95\% intervals paired over the rows
of day 7. One isotonic curve is compared with Platt scaling, one curve per segment with one
isotonic curve on the segments' rows. The model's AUC is $0.767$; uncalibrated, it predicts a
mean click probability of $0.67$ for a click rate of $0.26$.}
\label{tab:calibration}
\small
\setlength{\tabcolsep}{3pt}
\begin{tabular}{@{}lrr@{}}
\toprule
calibration & ECE & $\Delta$ log-loss \\
\midrule
none & 0.41 & \\
downsampling correction & 0.032 & \\
Platt scaling~\cite{platt1999} & 0.0080 & \\
isotonic, one curve & 0.0015 & $-5.0$ ($-5.3$ to $-4.7$) $\cdot 10^{-4}$ \\
\midrule
segments (60\% of rows), one curve & 0.0023 & \\
one curve per segment & 0.0026 & $-1.24$ ($-1.42$ to $-1.07$) $\cdot 10^{-3}$ \\
\bottomrule
\end{tabular}
\end{table}

\end{document}